\documentclass[11pt,letterpaper]{article}

\usepackage[margin=0.82in,headheight=14pt,headsep=10pt]{geometry}
\usepackage[T1]{fontenc}
\usepackage[utf8]{inputenc}
\usepackage{mathtools,amsthm}
\usepackage{newtxtext,newtxmath}
\usepackage{microtype}
\usepackage{booktabs,tabularx,longtable,array}
\usepackage{enumitem}
\usepackage{fancyhdr}
\usepackage{titlesec}
\usepackage{xcolor}
\usepackage{tikz}
\usetikzlibrary{arrows.meta,positioning}
\usepackage[ruled,vlined,linesnumbered]{algorithm2e}
\usepackage[most]{tcolorbox}
\usepackage{pdflscape}
\usepackage[square,authoryear]{natbib}
\usepackage{hyperref}

\definecolor{paperblue}{HTML}{173A5E}
\definecolor{accentblue}{HTML}{2D6A9F}
\definecolor{findingsbg}{HTML}{EAF2F8}

\hypersetup{
  colorlinks=true,
  linkcolor=paperblue,
  citecolor=paperblue,
  urlcolor=paperblue,
  pdftitle={SAT Certificates for the Matrix-Multiplication Challenges over F2},
  pdfauthor={Nick Palladinos}
}
\titleformat{\section}
  {\large\bfseries\color{paperblue}}{\thesection}{0.65em}{}
\titleformat{\subsection}
  {\normalsize\bfseries\color{paperblue}}{\thesubsection}{0.65em}{}
\titlespacing*{\section}{0pt}{1.35em}{0.65em}
\titlespacing*{\subsection}{0pt}{1.1em}{0.5em}

\setlist{nosep}
\setlist[enumerate,1]{leftmargin=2.4em}
\setlist[itemize,1]{leftmargin=1.7em}
\renewcommand{\arraystretch}{1.08}

\newtheorem{proposition}{Proposition}[section]
\newtheorem{lemma}[proposition]{Lemma}
\newtheorem{theorem}[proposition]{Theorem}

\newcommand{\F}{\mathbb{F}}
\newcommand{\GL}{\mathrm{GL}}
\newcommand{\model}{\texttt{.model}}

\newcommand{\benchmarkcommit}{%
  \href{https://github.com/marijnheule/matrix-challenges/commit/150b2e2f519fa9896ad6dbb5103d886aa009f872}%
       {\texttt{150b2e2f}}%
}
\newcommand{\artifactrepository}{%
  \url{https://github.com/palladin/sat-challenges}%
}
\newcommand{\artifactrevision}{%
  \href{https://github.com/palladin/sat-challenges/tree/41d912dcb1bcb3cde4ecfb70ffe4a4eb265b4027}%
       {\texttt{41d912dc}}%
}
\newcommand{\hashlines}[2]{%
  \shortstack[l]{\texttt{#1}\\[-0.15em]\texttt{#2}}%
}

\begin{document}

% --------------------------------------------------------------------------
% Front matter
% --------------------------------------------------------------------------
\begin{center}
  \vspace*{0.35em}
  {\LARGE\bfseries\color{paperblue}
    SAT Certificates for the Matrix-Multiplication\\
    Challenges over \(\F_2\)\par}
  \vspace{0.7em}
  {\Large
    All Ten ``Expected-UNSAT'' Instances Are Satisfiable,\\
    and a Type-3-Free Rank-23 Scheme\par}
  \vspace{1.1em}
  {\large Nick Palladinos\par}
  \vspace{0.45em}
  {30 July 2026
    \(\;\boldsymbol{\cdot}\;\) Version 1.0\par}
\end{center}

\vspace{1.7em}
\begin{center}
  \textbf{Abstract}
\end{center}
\begin{quote}
\small
The matrix-multiplication SAT benchmark of Heule, Kauers, and Seidl asks,
among other tasks, for solutions of ten known-satisfiable rank-23 formulas,
proofs of unsatisfiability for ten formulas expected to be unsatisfiable, and
a rank-23 scheme over \(\F_2\) having a summand with no type-3 monomial. We
give complete satisfying assignments for the 21 CNFs in the repository's
top-level \texttt{challenge1/}, \texttt{challenge2/}, and
\texttt{challenge3/} directories at commit \benchmarkcommit. The principal
finding is that all ten top-level Challenge-2 formulas are satisfiable. A
direct audit shows that their hardcoded type-3 pairings are imposed by
positive unit clauses on the 621 base variables: the formulas require
selected incidences but do not forbid additional type-3 incidences. Starting
from exact 23-summand schemes, we use the \(\GL(3,2)^3\) isotropy action,
cyclic trace symmetry, and perfect matching of transformed summands to
constrained slots to construct witnesses for all ten files. For Challenge 3,
we combine a locked semantic repair with a two-term identity over \(\F_2\) to
obtain a distinguished summand of type-3 count zero. Every semantic
decomposition has zero residual in all 729 Brent equations. The accompanying
DIMACS models assign all 26,541 variables of each formula and satisfy all
2,461,316 clauses across the 21 instances. A separate parser and clause
evaluator rechecks the emitted models. A deterministic one-file Python
reproducer regenerates the 21 certificates from the original CNFs in
approximately nine seconds on the reported test host.
\end{quote}


\medskip
\noindent\textbf{Artifact repository:} \artifactrepository\
(version-1.0 snapshot \artifactrevision).

\begin{tcolorbox}[
  enhanced,
  colback=findingsbg,
  colframe=accentblue,
  colbacktitle=accentblue,
  coltitle=white,
  fonttitle=\bfseries,
  title=Certificate-checked findings,
  boxrule=0.6pt,
  arc=2pt,
  left=7pt,right=7pt,top=5pt,bottom=5pt
]
\begin{tabularx}{\linewidth}{>{\bfseries}l X}
Challenge 1  & All 10 top-level satisfiable instances are solved from the
               CNFs by semantic local search.\\
Challenge 2  & All 10 top-level instances described as expected
               unsatisfiable are satisfiable as encoded.\\
Challenge 3  & The top-level CNF is satisfiable; its distinguished summand
               has type-3 count zero.\\
Verification & 21 total models, 557,361 assigned variables, 2,461,316
               checked clauses, zero failures.
\end{tabularx}
\end{tcolorbox}

% --------------------------------------------------------------------------
\section{Introduction}

A bilinear algorithm for multiplying two \(3\times3\) matrices with 23
scalar multiplications has been known since Laderman's construction
\citep{laderman1976}. Whether 22 multiplications suffice in the
noncommutative setting remains open; recent work still identifies 23 as the
lowest known multiplication count \citep{martensson2026}. One systematic
route to this small but difficult problem is to encode rank decompositions of the
matrix-multiplication tensor as polynomial equations and then as Boolean
satisfiability instances. Brent's equations provide the polynomial
formulation \citep{brent1970}; Courtois, Bard, and Hulme and later Heule,
Kauers, and Seidl demonstrated the effectiveness of SAT-based search for
rank-23 schemes \citep{courtois2011,heule2019a,heule2021}.

The benchmark repository accompanying \citet{heule2019a} contains four
challenge families \citep{heule2019b}. Challenge 1 consists of ten formulas
declared satisfiable but difficult for general SAT solvers. Challenge 2
consists of ten formulas with hardcoded pairings of type-3 terms that were
expected to be unsatisfiable. Challenge 3 asks whether a rank-23 scheme can
contain a summand with no type-3 term. Challenge 4 asks for a rank-22 scheme.

All repository-dependent claims in this paper refer to upstream commit
\benchmarkcommit. The audited scope is the 21 CNFs in its top-level
\texttt{challenge1/}, \texttt{challenge2/}, and \texttt{challenge3/}
directories; alternate collections elsewhere in the repository are outside
scope. The results do not address Challenge 4. The main contributions are:
\begin{enumerate}[label=(\roman*)]
  \item a compact semantic local-search procedure that solves all ten
        Challenge-1 instances while searching only the 621 factor variables;
  \item a clause-level audit explaining why the Challenge-2 formulas do not
        encode exact type-3 cores;
  \item explicit constructions, using tensor automorphisms and perfect
        matching, proving all ten Challenge-2 formulas satisfiable;
  \item an explicit \(\F_2\) decomposition into 23 rank-one summands, with
        one type-3-free summand satisfying the Challenge-3 CNF;
  \item complete DIMACS assignments and an end-to-end reproducer that
        independently evaluates every original clause.
\end{enumerate}

The key methodological distinction is between the semantic problem---729
parity equations in 621 meaningful bits---and the much larger Tseitin
encoding used by the CNF. Search and structural transformations are
performed semantically; auxiliary variables are reconstructed only after an
exact tensor decomposition is found.

% --------------------------------------------------------------------------
\section{Rank-23 schemes and the Boolean Brent equations}

\subsection{Tensor representation}

All indices range over \(\{0,1,2\}\). For each summand
\(t\in\{1,\ldots,23\}\), let
\[
  A_t,B_t,C_t\in\F_2^{3\times3}.
\]
The corresponding rank-one tensor is \(A_t\otimes B_t\otimes C_t\). The
matrices form a 23-summand decomposition of the \(3\times3\)
matrix-multiplication tensor when
\begin{equation}
  \bigoplus_{t=1}^{23}
  A_t(i_1,i_2)B_t(j_1,j_2)C_t(k_1,k_2)
  =
  \delta_{i_2,j_1}\delta_{j_2,k_1}\delta_{k_2,i_1}
  \label{eq:brent}
\end{equation}
for all six indices. These are the 729 cubic Brent equations in
\[
  23\cdot3\cdot9=621
\]
base variables \citep{brent1970,heule2019a}.

Equivalently, for input matrices \(X=(x_{ij})\) and \(Y=(y_{jk})\), define
\begin{equation}
  p_t=
  \left(\bigoplus_{i,j:A_t(i,j)=1}x_{ij}\right)
  \left(\bigoplus_{j,k:B_t(j,k)=1}y_{jk}\right).
  \label{eq:bilinear-product}
\end{equation}
Then XOR \(p_t\) into output \(z_{ik}\) whenever \(C_t(k,i)=1\).
Equation~\eqref{eq:brent} states that the resulting nine output bits are
exactly \(XY\) over \(\F_2\). The encoder stores all three coefficient
matrices as row-major 9-bit masks, except that the output mask uses bit
\(3k+i\) for \(z_{ik}\).

\begin{figure}[t]
\centering
\begin{tikzpicture}[
  process/.style={
    draw=accentblue,
    fill=findingsbg,
    rounded corners=2pt,
    minimum width=2.55cm,
    minimum height=1.05cm,
    align=center,
    font=\small
  },
  flow/.style={-{Latex[length=2mm]},draw=paperblue,semithick},
  node distance=0.45cm and 0.48cm
]
  \node[process] (cnf) {Original CNF};
  \node[process, right=of cnf] (extract) {Extract base-\\side constraints};
  \node[process, right=of extract] (construct)
    {Semantic search\\or symmetry\\construction};
  \node[process, below=of construct] (tensor)
    {Check all\\729 tensor\\coordinates};
  \node[process, left=of tensor] (propagate)
    {Unit-propagate\\25,920\\auxiliaries};
  \node[process, left=of propagate] (evaluate)
    {Reparse and\\evaluate\\every clause};
  \draw[flow] (cnf) -- (extract);
  \draw[flow] (extract) -- (construct);
  \draw[flow] (construct) -- (tensor);
  \draw[flow] (tensor) -- (propagate);
  \draw[flow] (propagate) -- (evaluate);
\end{tikzpicture}
\caption{The common pipeline. Only the semantic construction differs among
the three challenge families.}
\label{fig:pipeline}
\end{figure}
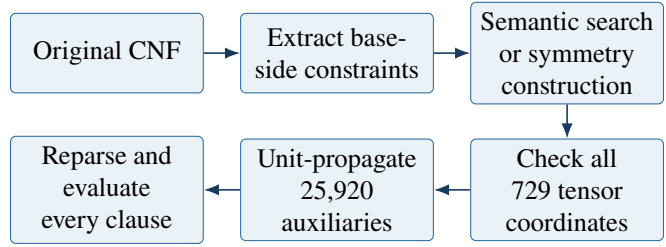

\subsection{Type of a monomial}

Following \citet{heule2019a}, the monomial
\(x_{i_1i_2}y_{j_1j_2}z_{k_2k_1}\) may be classified by
\begin{equation}
  m=\delta_{i_2,j_1}+\delta_{j_2,k_1}+\delta_{k_2,i_1}
  \in\{0,1,2,3\}.
  \label{eq:monomial-type}
\end{equation}
A type-3 monomial has all three index matches and is one of the 27 standard
matrix-multiplication monomials. For a summand \(t\), define its integer
type-3 count
\begin{equation}
  \tau_t=\sum_{i,j,k=0}^{2}A_t(i,j)B_t(j,k)C_t(k,i).
  \label{eq:type3-count}
\end{equation}
This count is not reduced modulo 2. Equation~\eqref{eq:brent} requires each
target monomial to occur with odd total parity, but it does not by itself
require a unique occurrence.

\subsection{CNF encoding and certificates}

The published formulas introduce Tseitin variables for conjunctions and XOR
chains \citep{tseitin1983,heule2019a}. Every instance treated here has
26,541 DIMACS variables. Variables 1 through 621 are the base variables in
the order
\[
  (t,A,0\ldots8),\qquad(t,B,0\ldots8),\qquad(t,C,0\ldots8),
\]
while the remaining 25,920 variables are encoding auxiliaries.

A total DIMACS assignment satisfying every clause is a direct SAT
certificate. The supplied models therefore settle satisfiability
independently of the construction used to find them. The semantic
representation remains valuable because it exposes the tensor geometry and
reduces the number of actively searched variables by a factor of more than
40.

% --------------------------------------------------------------------------
\section{End-to-end certificate pipeline}

A state is stored as a Boolean array of shape \(23\times3\times9\). Its
represented tensor is the XOR of the 23 outer products, and its residual is
that tensor XOR the target tensor. The residual is maintained incrementally
under one-bit factor flips.

After a zero-residual state satisfying the direct side constraints is
obtained, the first 621 DIMACS variables are assigned from the state.
Repeated unit propagation over the original clauses determines all remaining
variables in every instance. The resulting assignment is checked against
every clause. Finally, a separate parser reads the written model and reparses
the original CNF before evaluating all clauses again. This last pass is
intentionally independent of the semantic tensor code and of the first
DIMACS parser.

\begin{proposition}[Certificate criterion]
For each instance in the companion artifact repository
(Section~\ref{sec:artifacts}), the emitted \model\ file is a total assignment
of variables \(1,\ldots,26{,}541\) and satisfies every clause of the
corresponding original CNF.
\end{proposition}

\begin{proof}
The checker rejects conflicting literals, missing variables, malformed
clause termination, and incorrect header counts. It stops at the first
unsatisfied clause. All 21 files were accepted. Aggregate counts are given
in Table~\ref{tab:verification}, and per-file hashes are listed in
Appendix~\ref{app:hashes}. The generated manifest records the same values.
\end{proof}

\begin{algorithm}[t]
\caption{Two-phase semantic residual walk}
\label{alg:semantic-walk}
\DontPrintSemicolon
\SetKwFor{ForEach}{for}{do}{end}
\SetKwFor{While}{while}{do}{end}
Initialize \(x\in\F_2^{23\times3\times9}\) respecting forced bits\;
Compute residual
\(r=\mathcal{T}_{3,3,3}\oplus\sum_t A_t\otimes B_t\otimes C_t\)\;
\ForEach{\emph{phase} in \emph{direct-repair}, \emph{gain/activation}}{
  \If{\emph{phase} is \emph{gain/activation}}{
    restore the best state and residual retained from direct repair\;
  }
  \While{\(r\neq0\) and phase budget remains}{
    choose a violated coordinate \((a,b,c)\)\;
    enumerate the 69 aligned factor flips\;
    discard forced flips; in direct phase discard non-immediate moves\;
    compute make and break counts incrementally\;
    score by break (direct) or break\(-\)make (gain)\;
    sample a move with Boltzmann weight and update \(x,r\)\;
  }
}
\Return{\(x\) if \(r=0\)}\;
\end{algorithm}

% --------------------------------------------------------------------------
\section{Challenge 1: semantic local search}

The first challenge was designed to encourage local search without the
streamlining constraints used in the original discovery process
\citep{heule2019a}. Direct local search on all 26,541 CNF variables obscures
the multilinear structure. Instead, the reference implementation moves only
among the 621 factor bits.

At each step, choose a nonzero residual coordinate \((a,b,c)\). For every
summand \(t\), consider the 69 coordinate-aligned flips
\[
  A_t[a],\qquad B_t[b],\qquad C_t[c].
\]
Flipping \(A_t[a]\), for example, toggles precisely the tensor slice
\(\{a\}\times\operatorname{supp}(B_t)\times\operatorname{supp}(C_t)\).
This makes the exact numbers of corrected and damaged residual coordinates
inexpensive to compute.

The walk has two phases. During direct repair, a move is eligible only if it
immediately toggles the selected defect; it is scored by the number of
currently correct coordinates it breaks. During gain/activation, all 69
aligned moves are eligible, including moves that prepare a term for a later
repair, and the score is
\begin{equation}
  s=\operatorname{break}-\operatorname{make},
  \label{eq:walk-score}
\end{equation}
which is exactly the change in residual Hamming weight. Among candidates,
moves are sampled with weight
\begin{equation}
  \exp\!\left(-1.5(s-s_{\min})\right).
  \label{eq:boltzmann}
\end{equation}

Positive unit clauses on base variables are treated as immutable bits.
Initial free bits are independently set with probability \(1/8\), except for
one instance using \(1/10\).

The fixed schedule is 400,000 direct-repair iterations followed by 600,000
gain/activation iterations for the standalone Challenge-1 runs. At the phase
boundary, the implementation restores the best residual state retained
during direct repair. The second phase is essential: activation moves can
temporarily leave the chosen defect untouched while reshaping a rank-one
term into a coherent future cancellation.

All ten walks terminate at zero tensor residual; the generated total models
satisfy the original formulas. The semantic search is generic in the sense
that the target tensor and direct CNF side constraints determine the
objective. The per-instance presets are only deterministic seeds and initial
densities, not hardcoded solutions.

% --------------------------------------------------------------------------
\section{Challenge 2: all ten expected-UNSAT formulas are SAT}

Challenge 2 was described as ten 23-multiplication formulas with hardcoded
pairings of type-3 terms that were expected to be unsatisfiable
\citep{heule2019a,heule2019b}. The word \emph{expected} is decisive: the
files were conjectural subproblems, not certified UNSAT instances.

\begin{table}[t]
\centering
\caption{Challenge-1 reproduction. ``Support'' is the number of base
variables set to one; \(\sum_t\tau_t\) is the total integer type-3
incidence. Times are semantic-search times from the reported run and exclude
model completion. Instance names omit the common \texttt{MM-23-} prefix.}
\label{tab:challenge1}
\small
\begin{tabular}{lrrrrrr}
\toprule
Instance & Seed & Density & Iterations & Time (s) & Support & \(\sum_t\tau_t\)\\
\midrule
2-2-2-2-3 & 3  & \(1/8\)  & 180,423 & 0.284 & 175 & 29\\
2-2-2-2-A & 2  & \(1/8\)  &   4,418 & 0.002 & 146 & 27\\
2-2-2-2-B & 14 & \(1/8\)  & 156,123 & 0.105 & 162 & 27\\
2-2-2-2-C & 4  & \(1/8\)  &  40,264 & 0.019 & 166 & 27\\
2-2-2-2-D & 0  & \(1/8\)  &  11,977 & 0.007 & 153 & 27\\
2-2-2-2-M & 37 & \(1/8\)  & 529,421 & 0.598 & 159 & 27\\
2-2-2-3-4 & 3  & \(1/8\)  & 200,753 & 0.117 & 167 & 31\\
2-2-2-4-A & 20 & \(1/10\) & 587,039 & 0.608 & 179 & 29\\
2-2-2-4-B & 6  & \(1/8\)  &   6,869 & 0.004 & 155 & 29\\
4-4-4-4-1 & 0  & \(1/8\)  &   8,639 & 0.004 & 153 & 35\\
\bottomrule
\end{tabular}
\end{table}

\subsection{What the CNFs directly constrain}

We audited every clause whose variables all lie in the base range
\(1,\ldots,621\). For each Challenge-2 file:
\begin{itemize}
  \item every base-only clause is a positive unit clause;
  \item there are 77--81 distinct positive units, depending on the file;
  \item there are no negative base-only units; and
  \item there are no non-unit base-only clauses.
\end{itemize}
Here ``base-only'' excludes the mixed clauses defining Tseitin auxiliaries,
which are part of the common Brent encoding. The extra units force selected
entries of selected \(A_t,B_t,C_t\) masks to be present. They do not directly
force other entries to be absent.

\begin{proposition}[Inclusion rather than exactness]
The Challenge-2 side constraints encode required type-3 incidences, not exact
type-3 cores. A satisfying summand may contain additional type-3 monomials
beyond the advertised grouping, provided all Brent parities remain correct.
\end{proposition}

This distinction opens a large space of dense witnesses. Indeed, the ten
certificates below have total type-3 incidence between 41 and 93, rather than
exactly 27. Every total is odd, as required by adding even numbers of extra
occurrences to the 27 target parities.

\subsection{Tensor automorphisms}

Identify each 9-bit factor mask with a \(3\times3\) matrix. Let
\(P,Q,R\in\GL(3,2)\). Define
\begin{equation}
\begin{aligned}
  A_t' &= P^{\mathsf T}A_tQ^{-\mathsf T},\\
  B_t' &= Q^{\mathsf T}B_tR^{-\mathsf T},\\
  C_t' &= R^{\mathsf T}C_tP^{-\mathsf T}.
\end{aligned}
\label{eq:isotropy-action}
\end{equation}
The action is a concrete part of the isotropy of matrix multiplication
studied by \citet{degroote1978}.

\begin{lemma}
Transformation~\eqref{eq:isotropy-action} preserves the
matrix-multiplication tensor.
\end{lemma}

\begin{proof}
Let \(\langle U,V\rangle=\operatorname{tr}(U^{\mathsf T}V)\) be the
Frobenius pairing. For arbitrary \(X,Y,Z\in\F_2^{3\times3}\),
\begin{align*}
\langle A_t',X\rangle
 &=\langle A_t,PXQ^{-1}\rangle,\\
\langle B_t',Y\rangle
 &=\langle B_t,QYR^{-1}\rangle,\\
\langle C_t',Z\rangle
 &=\langle C_t,RZP^{-1}\rangle.
\end{align*}
If the original factors decompose
\(\mathcal{T}(X,Y,Z)=\operatorname{tr}(XYZ)\), the transformed factors
therefore evaluate to
\begin{align*}
\mathcal{T}(PXQ^{-1},QYR^{-1},RZP^{-1})
 &=\operatorname{tr}(PXYZP^{-1})\\
 &=\operatorname{tr}(XYZ).
\end{align*}
Thus the complete trilinear tensor, not merely a scalar statistic of each
summand, is preserved.
\end{proof}

The cyclic factor map \((A,B,C)\mapsto(B,C,A)\) is also valid because the
target trilinear form satisfies
\(\mathcal{T}(Z,X,Y)=\operatorname{tr}(ZXY)=\operatorname{tr}(XYZ)\).

\begin{table}[t]
\centering
\caption{Challenge-2 constructions and witness statistics. Instance names
omit \texttt{MM-23-}. ``Cyc.'' denotes the cyclic factor rotation. The
\(P,Q,R\) columns are zero-based indices in the deterministic \(\GL(3,2)\)
enumeration.}
\label{tab:challenge2}
\small
\begin{tabular}{lrrlrrrr}
\toprule
Instance & Source seed & Cyc. & \((P,Q,R)\) & Units & Support
  & \(\sum_t\tau_t\) & \(\max\tau_t\)\\
\midrule
2-2-3-A & 128 & no  & (117, 117, 155) & 81 & 318 & 93 & 10\\
2-2-3-B & 145 & yes & (167, 158, 119) & 81 & 282 & 67 & 15\\
2-4-A   & 87  & no  & (118, 155, 54)  & 81 & 281 & 65 & 11\\
2-4-B   & 15  & no  & (71, 31, 155)   & 81 & 300 & 71 & 15\\
3-3-A   & 49  & no  & (71, 143, 154)  & 81 & 312 & 87 & 10\\
3-3-B   & 8   & no  & (156, 67, 163)  & 81 & 292 & 67 & 9\\
3-3-M   & 15  & no  & (81, 106, 53)   & 77 & 216 & 41 & 6\\
3-3-N   & 106 & no  & (163, 119, 131) & 80 & 319 & 87 & 12\\
5-A     & 20  & no  & (79, 139, 161)  & 81 & 285 & 65 & 15\\
5-B     & 15  & no  & (141, 31, 143)  & 81 & 301 & 75 & 15\\
\bottomrule
\end{tabular}
\end{table}

Since \(\lvert\GL(3,2)\rvert=168\), the direct product contains
\(168^3=4{,}741{,}632\) triples. During discovery this finite action provided
a structured search space. The final reproducer does not repeat an
exhaustive search: it applies the fixed triples in Table~\ref{tab:challenge2}.

\subsection{Assigning transformed terms to constrained slots}

The 23 summands are unordered. After an optional cyclic factor rotation and
the action~\eqref{eq:isotropy-action}, construct a bipartite graph. The left
vertices are the 23 constrained target slots; the right vertices are the 23
transformed source terms. An edge exists when a transformed term contains
every base bit forced in that target slot. A perfect matching gives a
permutation of the terms satisfying all side constraints, while leaving
their tensor sum unchanged.

The implementation uses deterministic augmenting paths, processing the
tightest slots first. Source schemes are freshly found from
\texttt{MM-23-2-2-2-2-A.cnf} by the semantic walk with the seeds in
Table~\ref{tab:challenge2}. The group elements are indexed by enumerating all
9-bit row-major matrices in increasing bit-pattern order and retaining the
168 invertible matrices.

\begin{theorem}[Challenge-2 classification]
Every one of the ten published Challenge-2 CNFs is satisfiable.
\end{theorem}

\begin{proof}
For each row of Table~\ref{tab:challenge2}, the stated source seed, optional
cyclic rotation, \(\GL(3,2)^3\) triple, and deterministic perfect matching
produce a zero-residual semantic assignment satisfying all forced base
units. Unit propagation completes it to all 26,541 variables. The
accompanying total model then independently satisfies every original clause.
A single satisfying assignment is a definitive SAT certificate.
\end{proof}

The result settles the truth value of these exact CNF files. It does not
settle a stronger, differently encoded question in which the named type-3
grouping is required to be exact. Such variants would need negative
exclusions or equivalent exact-cardinality constraints.

% --------------------------------------------------------------------------
\section{Challenge 3: a type-3-free summand}

The Challenge-3 formula forbids every type-3 monomial in one distinguished
summand, with no streamlining constraints. At base-variable level this is
represented by the 27 clauses
\begin{equation}
  \neg A(i,j)\vee\neg B(j,k)\vee\neg C(k,i),
  \qquad i,j,k\in\{0,1,2\}.
  \label{eq:no-type3-clauses}
\end{equation}
The order of summands is immaterial; in the decomposition reported here, the
constrained summand is listed first.

\begin{table}[t]
\centering
\caption{Independent full-model verification. ``Assigned variables'' counts
variables across files, not distinct variable names globally.}
\label{tab:verification}
\begin{tabular}{lrrrr}
\toprule
Family & Models & Assigned variables & Clauses evaluated & Failed clauses\\
\midrule
Challenge 1 & 10 & 265,410 & 1,172,098 & 0\\
Challenge 2 & 10 & 265,410 & 1,172,065 & 0\\
Challenge 3 & 1  &  26,541 &   117,153 & 0\\
\midrule
\textbf{Total} & \textbf{21} & \textbf{557,361}
  & \textbf{2,461,316} & \textbf{0}\\
\bottomrule
\end{tabular}
\end{table}

\subsection{Locked repair and a rank-preserving identity}

Start from a freshly searched exact solution of
\texttt{MM-23-2-2-2-3-4.cnf}. Three terms are fixed and locked to the masks
\begin{equation}
  (64,4,511),\qquad(92,89,2),\qquad(438,24,2).
  \label{eq:locked-terms}
\end{equation}
The first mask triple is already present in the seed-3 source; the other two
are changed. Fixing all three leaves a tensor residual of Hamming weight 12.
Running the same semantic walk while forbidding changes to the locked terms
repairs the residual to zero. In the reproduced run, seed 14 required 10,918
iterations.

The last two terms in~\eqref{eq:locked-terms} share the output factor \(C=2\).
Over \(\F_2\),
\begin{equation}
  92\otimes89+438\otimes24
  =92\otimes65+490\otimes24,
  \label{eq:rank-preserving-identity}
\end{equation}
because \(89=65\oplus24\) and \(490=92\oplus438\). Explicitly,
\[
  92\otimes(65+24)+438\otimes24
  =92\otimes65+(92+438)\otimes24.
\]
Tensoring both sides with \(C=2\) replaces two rank-one summands by two
rank-one summands and preserves the complete tensor exactly. A final
permutation places \((490,24,2)\) in the constrained slot.

\begin{theorem}[Type-3-free 23-summand scheme]
There exists a decomposition of the \(3\times3\) matrix-multiplication
tensor over \(\F_2\) into 23 rank-one summands, one of which has type-3 count
zero. The decomposition in Appendix~\ref{app:challenge3} satisfies the
top-level Challenge-3 CNF at commit \benchmarkcommit.
\end{theorem}

\begin{proof}
The complete list of 23 masks has zero residual in all 729 equations. For
the first term,
\begin{align*}
A_1=490 &: x_{01}+x_{10}+x_{12}+x_{20}+x_{21}+x_{22},\\
B_1=24  &: y_{10}+y_{11},\\
C_1=2   &: z_{10}.
\end{align*}
The potential standard monomial through \(y_{10}\) and \(z_{10}\) would
require \(x_{11}\), which is absent. The variable \(y_{11}\) has
output-column index 1, whereas the only selected output is column 0. Hence
Equation~\eqref{eq:type3-count} gives \(\tau_1=0\). The full DIMACS model
satisfies all 117,153 clauses of the original file.
\end{proof}

This construction is more informative than an opaque SAT assignment: the
short identity~\eqref{eq:rank-preserving-identity} isolates the structural
operation that creates the forbidden-free summand while keeping the
decomposition length at 23.

% --------------------------------------------------------------------------
\section{Verification and reproducibility}

\subsection{Aggregate verification}

Every semantic state also passes a direct 729-coordinate tensor check.
Challenge 1 has no zero-type summands in the supplied witnesses. Challenge 2
has no zero-type summands but substantially more than 27 total type-3
incidences. Challenge 3 has exactly one zero-type summand, with total type-3
incidence 31.

\subsection{Reproduction command and environment}

The one-file reproducer requires only the original benchmark tree, NumPy,
and Numba:
\begin{verbatim}
git -C /path/to/matrix-challenges checkout --detach \
  150b2e2f519fa9896ad6dbb5103d886aa009f872
python -m pip install numpy==2.3.5 numba==0.65.1
python mm23_all_reproduce.py /path/to/matrix-challenges -o models
\end{verbatim}
It does not read saved semantic states, checkpoints, or precomputed DIMACS
models. It searches the Challenge-1 source schemes, performs the fixed
Challenge-2 transformations and matchings, constructs the Challenge-3
witness, completes all auxiliary assignments, and verifies every written
model.

An author-reported full run used Python 3.13.5, NumPy 2.3.5, Numba 0.65.1,
and five virtual CPU cores reported as an AMD EPYC 9V74. The program reported
7.872 seconds; process wall time was 8.95 seconds with a maximum resident set
size of approximately 514 MiB. No raw log from that run is retained. The
measurements include Numba compilation, CNF parsing, semantic construction,
propagation, complete clause evaluation, and file output. They are intended
as reproducibility data rather than a comparative performance benchmark.

\subsection{Artifact availability and hashes}
\label{sec:artifacts}

The code, complete DIMACS models, paper source, and supporting decompositions
are publicly available in the companion artifact repository
\citep{palladinos2026artifacts}: \artifactrepository. The version-1.0
artifact used for this paper is fixed by revision \artifactrevision. Its
principal files are:
\begin{itemize}
  \item \texttt{mm23\_all\_reproduce.py}: the end-to-end reproducer;
  \item \texttt{mm23\_complete\_reproduction/}: 21 \model\
        certificates and 21 human-readable \texttt{.md} decompositions;
  \item \texttt{paper/}: the editable \LaTeX\ source, references, and compiled
        paper.
\end{itemize}
For each CNF, the reproducer emits a \model\ certificate and a structured
\texttt{.terms.json} decomposition; it also emits \texttt{manifest.json},
containing per-instance dimensions, semantic statistics, CNF hashes, and
model hashes. The generated \texttt{.terms.json} files and
\texttt{manifest.json} are written to the selected output directory and are
not checked in. The distributed \texttt{.md} files are human-readable
presentations of the same 23 mask triples, not additional reproducer output.
The SHA-256 digest of the accompanying reproducer is
\begin{center}
\small
\begin{tabular}{r@{\;:\;}l}
\texttt{mm23\_all\_reproduce.py}
  & \hashlines{f06891802f0a7f3e4f3ba7314b813ef5}
              {67f3a189a603a24fa8200088901a6588}
\end{tabular}
\end{center}
The complete per-model SHA-256 list appears in Appendix~\ref{app:hashes}.

% --------------------------------------------------------------------------
\section{Discussion}

\subsection{Why the Challenge-2 outcome is plausible in retrospect}

A filename such as \texttt{MM-23-2-2-3-A.cnf} naturally suggests an exact
combinatorial core. The actual Boolean formula is the authoritative object,
however. Positive units assert that chosen type-3 incidences occur, while the
parity equations permit additional incidences in even multiplicities. Dense
orbit representatives therefore have room to meet the required units
without preserving the sparse intended pattern.

The result does not contradict the reported weakness of complete SAT solvers
on these encodings. Generic CDCL search can still perform poorly even when a
structured satisfying assignment exists. The successful route exploits
information that a generic solver does not automatically expose: the
621-variable tensor semantics, the isotropy group, cyclic trace invariance,
and the irrelevance of summand order.

\subsection{Implications for benchmark design}

When a benchmark is intended to test the extendability of an exact core, the
CNF should distinguish at least three notions:
\begin{enumerate}
  \item required incidences;
  \item forbidden incidences; and
  \item parity constraints implied by the tensor equations.
\end{enumerate}
For exact core patterns, positive units alone are insufficient. Negative unit
clauses, support equalities, or exact-cardinality constraints can close the
gap. Publishing a small side-constraint manifest alongside the generated
CNF would also make the intended semantics easier to audit.

\subsection{Scope and limitations}

The certificates establish statements over \(\F_2\) and settle the exact
supplied CNF files. They do not automatically provide noncommutative
algorithms over arbitrary coefficient rings; lifting a Boolean scheme is a
separate algebraic problem \citep{heule2021}. The certificates make no
determination about rank 22, the unsatisfiability of exact-core variants, or
equivalence classes of the 21 schemes. We also make no historical priority
claim beyond reporting the accompanying, independently checkable artifacts.

The stochastic component is made deterministic by fixed seeds, but a
different numerical or compiler environment could in principle alter a
random trajectory. This does not affect certificate validity: the supplied
total assignments can always be checked directly, and the construction
metadata gives a second route to the Challenge-2 witnesses.

% --------------------------------------------------------------------------
\section{Conclusion}

The published matrix-multiplication benchmark contains a sharper and more
surprising result than its original Challenge-2 expectation suggested. All
ten proposed UNSAT subproblems are SAT as encoded. The decisive observations
are that their direct core constraints are positive inclusions and that exact
23-summand schemes can be transported through a large, explicit
tensor-symmetry group before their summands are matched to constrained slots.
The same semantic viewpoint solves all ten Challenge-1 formulas and, combined
with a small rank-preserving identity, yields a Challenge-3 scheme with a
type-3-free summand.

The final evidence is deliberately simple: 21 complete assignments, each
assigning all 26,541 DIMACS variables, and zero failed clauses across
2,461,316 evaluations. These certificates resolve the 21 top-level
Challenge-1--3 CNFs at commit \benchmarkcommit\ and leave the rank-22
Challenge 4 as the remaining benchmark frontier.

\section*{Acknowledgments}

This work would not have been possible without the computational exploration
provided by OpenAI GPT-5.6 Sol. The mathematical and SAT claims in this paper
are grounded in the accompanying independently checkable models. The author
thanks Marijn J. H. Heule, Manuel Kauers, and Martina Seidl for creating and
publishing the benchmark and its encoding.

% --------------------------------------------------------------------------
% Appendices
% --------------------------------------------------------------------------
\clearpage
\appendix

\begin{landscape}
\section{Complete Challenge-3 decomposition}
\label{app:challenge3}

For each row, compute \(p_t=\alpha_t(X)\beta_t(Y)\) and XOR it into every
listed output. The output-mask convention is
\(C_t(k,i)\leftrightarrow z_{ik}\). The masks are decimal encodings of the
row-major 9-bit coefficient matrices.

\begingroup
\scriptsize
\setlength{\tabcolsep}{3.5pt}
\renewcommand{\arraystretch}{1.18}
\begin{longtable}{
  r r r r
  >{\raggedright\arraybackslash}p{3.25cm}
  >{\raggedright\arraybackslash}p{3.05cm}
  >{\raggedright\arraybackslash}p{3.55cm}
  r
}
\caption{The 23-summand Challenge-3 scheme.}
\label{tab:challenge3-scheme}\\
\toprule
\(t\) & \(A\) & \(B\) & \(C\) & \(\alpha_t(X)\) & \(\beta_t(Y)\)
  & Outputs & \(\tau_t\)\\
\midrule
\endfirsthead
\multicolumn{8}{c}{\itshape Table \thetable\ continued}\\
\toprule
\(t\) & \(A\) & \(B\) & \(C\) & \(\alpha_t(X)\) & \(\beta_t(Y)\)
  & Outputs & \(\tau_t\)\\
\midrule
\endhead
\bottomrule
\endfoot
1 & 490 & 24 & 2
 & \(x_{01}+x_{10}+x_{12}+x_{20}+x_{21}+x_{22}\)
 & \(y_{10}+y_{11}\)
 & \(z_{10}\) & 0\\
2 & 384 & 64 & 262
 & \(x_{21}+x_{22}\)
 & \(y_{20}\)
 & \(z_{10}+z_{20}+z_{22}\) & 1\\
3 & 128 & 72 & 39
 & \(x_{21}\)
 & \(y_{10}+y_{20}\)
 & \(z_{00}+z_{10}+z_{20}+z_{21}\) & 1\\
4 & 16 & 432 & 24
 & \(x_{11}\)
 & \(y_{11}+y_{12}+y_{21}+y_{22}\)
 & \(z_{01}+z_{11}\) & 1\\
5 & 92 & 130 & 144
 & \(x_{02}+x_{10}+x_{11}+x_{20}\)
 & \(y_{01}+y_{21}\)
 & \(z_{11}+z_{12}\) & 1\\
6 & 92 & 65 & 2
 & \(x_{02}+x_{10}+x_{11}+x_{20}\)
 & \(y_{00}+y_{20}\)
 & \(z_{10}\) & 1\\
7 & 56 & 6 & 48
 & \(x_{10}+x_{11}+x_{12}\)
 & \(y_{01}+y_{02}\)
 & \(z_{11}+z_{21}\) & 1\\
8 & 120 & 134 & 176
 & \(x_{10}+x_{11}+x_{12}+x_{20}\)
 & \(y_{01}+y_{02}+y_{21}\)
 & \(z_{11}+z_{21}+z_{12}\) & 4\\
9 & 376 & 216 & 32
 & \(x_{10}+x_{11}+x_{12}+x_{20}+x_{22}\)
 & \(y_{10}+y_{11}+y_{20}+y_{21}\)
 & \(z_{21}\) & 1\\
10 & 504 & 88 & 34
 & \(x_{10}+x_{11}+x_{12}+x_{20}+x_{21}+x_{22}\)
 & \(y_{10}+y_{11}+y_{20}\)
 & \(z_{10}+z_{21}\) & 3\\
11 & 64 & 4 & 511
 & \(x_{20}\)
 & \(y_{02}\)
 & \(z_{00}+z_{10}+z_{20}+z_{01}+z_{11}+z_{21}+z_{02}+z_{12}+z_{22}\)
 & 1\\
12 & 4 & 288 & 472
 & \(x_{02}\)
 & \(y_{12}+y_{22}\)
 & \(z_{01}+z_{11}+z_{02}+z_{12}+z_{22}\) & 1\\
\pagebreak
13 & 260 & 320 & 256
 & \(x_{02}+x_{22}\)
 & \(y_{20}+y_{22}\)
 & \(z_{22}\) & 1\\
14 & 132 & 96 & 259
 & \(x_{02}+x_{21}\)
 & \(y_{12}+y_{20}\)
 & \(z_{00}+z_{10}+z_{22}\) & 2\\
15 & 130 & 40 & 3
 & \(x_{01}+x_{21}\)
 & \(y_{10}+y_{12}\)
 & \(z_{00}+z_{10}\) & 1\\
16 & 18 & 17 & 10
 & \(x_{01}+x_{11}\)
 & \(y_{00}+y_{11}\)
 & \(z_{10}+z_{01}\) & 1\\
17 & 6 & 33 & 67
 & \(x_{01}+x_{02}\)
 & \(y_{00}+y_{12}\)
 & \(z_{00}+z_{10}+z_{02}\) & 1\\
18 & 20 & 418 & 152
 & \(x_{02}+x_{11}\)
 & \(y_{01}+y_{12}+y_{21}+y_{22}\)
 & \(z_{01}+z_{11}+z_{12}\) & 2\\
19 & 48 & 390 & 128
 & \(x_{11}+x_{12}\)
 & \(y_{01}+y_{02}+y_{21}+y_{22}\)
 & \(z_{12}\) & 1\\
20 & 1 & 5 & 15
 & \(x_{00}\)
 & \(y_{00}+y_{02}\)
 & \(z_{00}+z_{10}+z_{20}+z_{01}\) & 1\\
21 & 71 & 1 & 73
 & \(x_{00}+x_{01}+x_{02}+x_{20}\)
 & \(y_{00}\)
 & \(z_{00}+z_{01}+z_{02}\) & 1\\
22 & 65 & 5 & 79
 & \(x_{00}+x_{20}\)
 & \(y_{00}+y_{02}\)
 & \(z_{00}+z_{10}+z_{20}+z_{01}+z_{02}\) & 3\\
23 & 21 & 3 & 8
 & \(x_{00}+x_{02}+x_{11}\)
 & \(y_{00}+y_{01}\)
 & \(z_{01}\) & 1\\
\end{longtable}
\endgroup
\end{landscape}

\section{Challenge-2 certificate metadata}
\label{app:challenge2-metadata}

Each Challenge-2 model is reproduced from a semantic solution of
\texttt{MM-23-2-2-2-2-A.cnf}. The construction table in
Table~\ref{tab:challenge2} is sufficient together with the deterministic
definitions in the reference script: source search uses density \(1/8\) and
budgets 150,000 direct plus 350,000 gain steps; the optional outer map is
either identity or \((A,B,C)\mapsto(B,C,A)\); \(\GL(3,2)\) is enumerated by
increasing 9-bit row-major encoding; and target terms are matched by
deterministic augmenting paths. The resulting target-slot permutations are
recorded in the script-generated manifest.

\section{Complete model SHA-256 digests}
\label{app:hashes}

The following hashes identify the exact total assignments distributed in the
companion artifact repository.

\begin{center}
\footnotesize
\begin{tabular}{l l}
\toprule
Instance & Model SHA-256\\
\midrule
\texttt{MM-23-2-2-2-2-3}
  & \hashlines{2f3b4a8a48c405794588eceea5d83ad9}
              {8fc778e5c9a1b4413b17fdc5ea079aea}\\
\texttt{MM-23-2-2-2-2-A}
  & \hashlines{a2c9cd2f739867eac8f4df65e3fae5c1}
              {4a0ed4e8c8485d212879b24306405441}\\
\texttt{MM-23-2-2-2-2-B}
  & \hashlines{5d8e5886630630cd5d441f4aaca9a1c3}
              {bc02c8f53410327384f6a625af291926}\\
\texttt{MM-23-2-2-2-2-C}
  & \hashlines{6e9f960096087129a5ac023053a5bf25}
              {61a25ed06adb3e886b0b61fb77c9f5d0}\\
\texttt{MM-23-2-2-2-2-D}
  & \hashlines{8609df116787cdbadb1f3a29ca2e45f3}
              {d0a0d848097bfd4705e53b6ad0f03349}\\
\texttt{MM-23-2-2-2-2-M}
  & \hashlines{0e2011467202b739b381474b3bcdd21d}
              {5280d3ee8205641030d1cdadf8c45554}\\
\texttt{MM-23-2-2-2-3-4}
  & \hashlines{e13b3e3a19ac8204ba5282faab75d47e}
              {a097d9752d7ae29b25efeeb7ae34132b}\\
\texttt{MM-23-2-2-2-4-A}
  & \hashlines{0525d874f740cac01fb992197d554a37}
              {be11ad4bea8409218ea8d666dbc295d1}\\
\texttt{MM-23-2-2-2-4-B}
  & \hashlines{5d86abfe9cb889beb0861f02232484fd}
              {69ca01e0a3e8c440f78074a3b1a17c81}\\
\texttt{MM-23-4-4-4-4-1}
  & \hashlines{42d7c0b6ea1ca04d6b131e82e206ad7e}
              {09c9023687e85eb944c48eb77c6433e3}\\
\texttt{MM-23-2-2-3-A}
  & \hashlines{37bfafa7e8e81057a765f33ec1485e99}
              {c7a44386bd56b14fbd49ef555de6884e}\\
\texttt{MM-23-2-2-3-B}
  & \hashlines{de0d969b6b76464341f5ef5ba4246aac}
              {1ca66bf2ed29a725a41863f020d4b1d8}\\
\texttt{MM-23-2-4-A}
  & \hashlines{a7d4f2176c55f5a7e02d9bf15414c668}
              {45338da0855626c403c448fac34c31c0}\\
\texttt{MM-23-2-4-B}
  & \hashlines{082032bcd7092ba8df70e187720a4827}
              {10775e51fbcfe9084e5ed55ae759709c}\\
\texttt{MM-23-3-3-A}
  & \hashlines{365899d6bb12175b5ed866db4abc3d7e}
              {d6f028134becfa445442d09eff71903d}\\
\texttt{MM-23-3-3-B}
  & \hashlines{50cde7ea65310efa0eeefab326a28569}
              {4645e7da4b079ec5ad8d2fcbbbf1397f}\\
\texttt{MM-23-3-3-M}
  & \hashlines{6fcac8ded7bc52803303b0a934280318}
              {a3cf2d67ae2705606622a8dbb702523b}\\
\texttt{MM-23-3-3-N}
  & \hashlines{e91115ee8ce0e1970b81bcca37ec7588}
              {84cf5343a21f18203cd2f0b0611e8c7b}\\
\texttt{MM-23-5-A}
  & \hashlines{00b236cc65fdf3a3f8f794fa2b5b6a8f}
              {b727759622f798a03bc930150a9167ba}\\
\texttt{MM-23-5-B}
  & \hashlines{eff4da53e1dfd0f9464d2b7a2a88bca6}
              {503de103fe00b1011d12b748d7ebe9c6}\\
\texttt{MM-23-no-type3}
  & \hashlines{4c91a268d12675d3233e2f71d086ae96}
              {ab262154efdceb95609e5c59cfa08351}\\
\bottomrule
\end{tabular}
\end{center}

% --------------------------------------------------------------------------
% References are kept in this file so that a single editable source compiles
% with pdfLaTeX alone.  The optional paper/references.bib contains the same
% records for users who prefer a BibTeX workflow.
% --------------------------------------------------------------------------

\end{document}